\documentclass{article}

\usepackage{amsmath}
\usepackage{amssymb}
\usepackage{amsthm}
\usepackage[hidelinks, hypertexnames=false]{hyperref}
\usepackage{tikz}
\usepackage{appendix}

\newtheorem{Theorem}{Theorem}

\numberwithin{Theorem}{section}

\begin{document}

%

\begin{center}
\textbf{\large NEW SOLUTIONS WITH PEAKON CREATION IN\\
THE CAMASSA--HOLM AND NOVIKOV EQUATIONS}
\end{center}

\begin{center}
{\footnotesize M. KARDELL}
\end{center}

\vphantom{Text}

\begin{abstract}
In this article we study a new kind of unbounded solutions to the Novikov equation, found via a Lie symmetry analysis. These solutions exhibit peakon creation, i.e., these solutions are smooth up until a certain finite time, at which a peak is created. We show that the functions are still weak solutions for those times where the peak lives. We also find similar unbounded solutions with peakon creation in the related Camassa--Holm equation, by making an ansatz inspired by the Novikov solutions. Finally, we see that the same ansatz for the Degasperis--Procesi equation yields unbounded solutions where a peakon is present for all times.\end{abstract}


\section{Introduction}

In 1993, Camassa and Holm \cite{CHOriginal} discovered an integrable partial differential equation within the context of shallow water theory, an equation which has since been studied quite extensively. One reason for the interest in this equation is that it allows (weak) explicit solutions in the form of so called multipeakons. More recent equations with similar properties include the Degasperis--Procesi \cite{DPOriginal} and the Novikov \cite{Novikov} equations. 

The results of this article originated from a Lie symmetry analysis of the Novikov equation. This framework gives a complete list of transformations such that each solution of the equation is mapped to another solution. In the resulting list of transformations, there are two nontrivial transformations which we use to produce new solutions to the Novikov equation.

In fact, applying the new transformations found in this article to the Novikov one-peakon solution gives an unbounded solution displaying quite interesting behaviour. We find that this solution depends smoothly on \(x\) for some interval in time, and has peakon creation (or destruction, depending on the transformation) at some finite time \(t\). We also show that these functions are still weak solutions for those times for which the peak lives.

By making an ansatz inspired by the Novikov solutions with peakon creation, we also find such solutions to the Camassa--Holm equation. It is interesting to note that, apparently, these solutions cannot be found using Camassa--Holm symmetries. Another thing to note is that the same ansatz does not give peakon creation in the closely related Degasperis--Procesi equation, instead we find a kind of unbounded peakon solution where the peak lives for all times.

\section{Novikov Solutions with Peakon Creation}

The Novikov equation, given by 
\begin{equation}
\label{eq:Novikov}
u_t - u_{xxt} = -4u^2u_x + 3uu_xu_{xx} + u^2u_{xxx},
\end{equation}
admits multi-peakon solutions 
\begin{equation}
u(x, t) = \sum_{i=1}^n m_i(t)e^{-|x - x_i(t)|}
\end{equation}
in a weak sense. The word peakon is short for `peaked soliton', where peaked means that there is some point where the left and right derivatives do not coincide. The peakons interact in quite a complicated way; see \cite{NovikovPeakons} for explicit time dependence of the functions \(\{x_i(t), m_i(t)\}\) and a weak formulation of the problem.

Consider the one-peakon solution \(u(x, t) = ce^{-\left|x - c^2t\right|}\). This is a peakon traveling to the right, with constant speed equal to the square of the height of the peakon (which differs from Camassa--Holm and Degasperis--Procesi peakons, where the speed is just equal to the height). For fixed \(t\), the peakon looks as in Figure \ref{fig:Peakon}.

\begin{figure}
\begin{center}
\begin{tikzpicture}
\draw[->] (-4, 0) -- (4.1, 0) node[right] {\(x\)}; 
\draw[->] (0, -0.5) -- (0, 1.5) node[above] {\(u(x, t)\)};
\draw [domain = -4.0 : 0.8] plot (\x, {exp(\x - 0.8)}); 
\draw [domain = 0.8 : 4.0] plot (\x, {exp(-\x + 0.8)}); 
\draw (0.8, 0.1) -- (0.8, -0.1) node[below] {\(c^2t\)};
\draw (0.1, 1) -- (-0.1, 1) node[left] {\(c\)};
\end{tikzpicture}
\caption{One-peakon solution}
\label{fig:Peakon}
\end{center}
\end{figure}
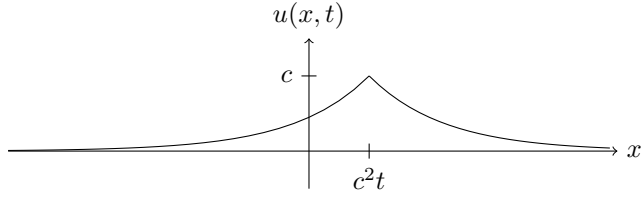

In the Appendix, \hyperref[Th:3]{Theorem \ref{Th:3}}, we compute the Lie symmetries of the Novikov equation. These correspond to transformations that take known (strong) solutions of the equation to other solutions. We repeat here the result for convenience.
\begin{Theorem}
If \(u = f(x, t)\) solves the Novikov equation (\ref{eq:Novikov}), then so do
\begin{gather*}
u_1 = f(x - \varepsilon, t), \\
u_2 = f(x, t - \varepsilon), \\
u_3 = e^{\varepsilon/2} f(x, te^{\varepsilon}), \\
u_4 = \sqrt{1 + 2\varepsilon e^{2x}} f\left(-\frac{1}{2} \textnormal{ln}\left(e^{-2x} +2\varepsilon\right), t\right), \\
u_5 = \sqrt{1 + 2\varepsilon e^{-2x}} f\left(\frac{1}{2} \textnormal{ln}\left(e^{2x} +2\varepsilon\right), t\right).\end{gather*}
\end{Theorem} 
In this section we study the functions that one gets by transforming the one-peakon solution. Note though, that the one-peakon is not a smooth solution, so we can not say \emph{a priori} whether this approach gives valid weak solutions of the Novikov equation, this has to be checked. Applying the first three transformations gives us translations and scaling of a peakon, hence no essentially new solutions come up. The fourth and fifth tranformations are more interesting. They give the functions
\begin{subequations}
\begin{gather}
u_4(x,t) = c\sqrt{1 + 2\varepsilon e^{2x}} e^{-\left|\frac{1}{2} \textnormal{ln}\left(e^{-2x} + 2\varepsilon\right) + c^2t\right|},\\
u_5(x,t) = c\sqrt{1 + 2\varepsilon e^{-2x}} e^{-\left|\frac{1}{2} \textnormal{ln}\left(e^{2x} + 2\varepsilon\right) - c^2t\right|}. \label{NovikovTransformedPeakonB}
\end{gather}
\end{subequations}
Note that these solutions do not tend to zero as \(|x| \to \infty\). Let us first study the function \(u_5(x,t)\).

\begin{Theorem} \label{NovikovPeakonCreationSolution}
The transformed Novikov peakon \[u_5(x,t) = c\sqrt{1 + 2\varepsilon e^{-2x}} e^{-\left|\frac{1}{2} \textnormal{ln}\left(e^{2x} + 2\varepsilon\right) - c^2t\right|}\] is a smooth solution to the Novikov equation up until \(t_0 = \frac{1}{2c^2}\textnormal{ln}(2\varepsilon)\), when a peak is created at \(x = -\infty\). After time \(t_0\), the function is still a weak solution.
\end{Theorem}

\begin{proof}
Let us examine the expression inside the modulus signs in \(u_5\). This expression is increasing in \(x\), and has the only root \(x = \frac{1}{2} \textnormal{ln} (e^{2c^2t} - 2\varepsilon)\). Thus, there can exist a value of \(x\) for which the expression changes sign, but only when \(t > t_0: = \frac{1}{2c^2}\textnormal{ln}(2\varepsilon)\). Before time \(t_0\), the function \(u_5\) is smooth, and is thus a solution of the Novikov equation in the usual sense. At the time \(t_0\) a peak (a point where the left and right derivatives are unequal) is created at \(x = -\infty\), which then moves in rapidly from the left. 

More concretely, for \(t \leq t_0\), the expression (\ref{NovikovTransformedPeakonB}) simplifies significantly, since
\[u_5(x,t) = c\sqrt{1 + 2\varepsilon e^{-2x}} e^{-\frac{1}{2} \textnormal{ln}\left(e^{2x} + 2\varepsilon\right) + c^2t} = c\frac{\sqrt{1 + 2\varepsilon e^{-2x}}}{\sqrt{e^{2x} + 2\varepsilon}} e^{c^2t} = ce^{-x + c^2t}.\]
For \(t > t_0\), one can simplify in a similar manner, depending on whether one is to the left or to the right of the peak at \(B(t) := \frac{1}{2} \textnormal{ln} (e^{2c^2t} - 2\varepsilon)\), yielding
\begin{gather}
u_5(x,t) = 
\begin{cases}
ce^{-x + c^2t}, &x \geq B(t) \\
c(e^x + 2\varepsilon e^{-x}) e^{-c^2t}. &x \leq B(t)
\end{cases}
\label{NovikovSolutonU5}
\end{gather}
To check that a function is still a weak solution after time \(t_0\), in the sense of \cite{NovikovPeakons}, one needs to show that
\[\left<\left(1 - \partial_x^2\right)u_t + \left(4 - \partial_x^2\right)\partial_x\left(\frac{1}{3}u^3\right) + \partial_x\left(\frac{3}{2}uu_x^2\right) + \frac{1}{2}u_x^3, \phi\right> = 0, \quad \forall \phi(x) \in C_0^{\infty},\]
where \(\left<\cdot, \cdot\right>\) means action on test functions in the usual sense. Using the definition of distributional derivatives, one gets
\begin{gather}
\left<u_t, \left(1 - \partial_x^2 \right) \phi \right> + \left<\frac{1}{3}u^3, \partial_x\left(\partial_x^2 - 4\right)\phi \right> + \left<\frac{3}{2}uu_x^2, -\partial_x \phi \right> + \left<\frac{1}{2}u_x^3, \phi \right> = 0.  \label{NovikovCheckWeak} \end{gather}
Let \(u^+\) and \(u^-\) be the expressions of (\ref{NovikovSolutonU5}) to the right and left of the peak, respectively. Note that \(u_5(x, t)\) is continuous at all points, with \(u_x\) and \(u_t\) piecewise continuous functions, so the lefthand side in (\ref{NovikovCheckWeak}) equals
\[\begin{split}
&\int_B^{\infty} \! u_t^+ \left(\phi - \phi_{xx}\right) \, dx + \int_{-\infty}^B \! u_t^- \left(\phi - \phi_{xx}\right) \, dx + \int_B^{\infty} \frac{1}{3}\left(u^+\right)^3 \left(\phi_{xxx} - 4\phi_x\right) \, dx \, + \\ 
&\quad+ \int_{-\infty}^B \frac{1}{3}\left(u^-\right)^3  \left(\phi_{xxx} - 4\phi_x\right) \, dx  + \int_B^{\infty} \frac{3}{2}u^+ \left(u_x^+\right)^2 (-\phi_x) \, dx \\ &\quad+ \int_{-\infty}^B \frac{3}{2}u^- \left(u_x^-\right)^2 (-\phi_x) \, dx \, +
+ \int_B^{\infty} \frac{1}{2}\left(u_x^+\right)^3 \phi \, dx + \int_{-\infty}^B \frac{1}{2}\left(u_x^-\right)^3 \phi \, dx.
\end{split}\]
Using integration by parts to move the derivatives back to \(u\), we get two kinds of terms. First we again get integrals, which combine to zero since \(u\) is a strong solution of the Novikov equation on each interval. The boundary values at infinity are all zero, since we integrate against a test function with compact support, but we also get boundary values at \(B\):
\begin{gather} U_1(B)\phi(B) + U_2(B)\phi_x(B) + U_3(B)\phi_{xx}(B), \end{gather}
where we use the shorthand notation \(f(B) = f(B(t), t)\), and
\begin{align*} U_1(B) :=& \left(u_t^-\right)_x(B) - \left(u_t^+\right)_x(B) + \frac{1}{3}\left(\left(u^-\right)^3\right)_{xx}(B) - \frac{1}{3}\left(\left(u^+\right)^3\right)_{xx}(B)\\
& + \frac{3}{2}u^+(B)(u_x^+(B))^2 - \frac{3}{2}u^-(B)(u_x^-(B))^2 + \frac{4}{3}(u^+)^3(B) - \frac{4}{3}(u^-)^3(B), \\
U_2(B) :=& u_t^+(B) - u_t^-(B) + \frac{1}{3}\left(\left(u^+\right)^3\right)_x(B) - \frac{1}{3}\left(\left(u^-\right)^3\right)_x(B), \\
U_3(B) :=& \frac{1}{3}\left(\left(u^-\right)^3\right)(B) - \frac{1}{3}\left(\left(u^+\right)^3\right)(B). \end{align*}
The continuity of \(u_5\) gives \(u^+(B) = u^-(B)\) which means that \(U_3(B)\) is zero. It is not obvious, but easy to check with computer, that \(U_1(B)\) and \(U_2(B)\) are also zero. For example,
\begin{equation}
\left(u_t + \frac{1}{3}\left(u^3\right)_x\right)\left(B\right) = \frac{-2\varepsilon c^3 e^{c^2t}}{\left(e^{2c^2t} - 2\varepsilon\right)^{\frac{3}{2}}}
\end{equation}
for both \(u^+\) and \(u^-\), showing that \(U_2(B) = 0\).
\end{proof}

Note that as the peak moves in from the left, it is not actually a local maximum from the start (so it might be more accurate to call it a corner), as we can see from Figure \ref{fig:PeakonCreation1}. As time increases the corner really turns into a peak, indicated in Figure \ref{fig:PeakonCreation2}. The peak becomes increasingly separated from the large wave to the left, and one can see from the expression for \(B(t)\) that, asymptotically, the peak moves to the right with constant speed \(c^2t\) like a one-peakon solution, unaffected by the wavefront. Figure \ref{fig:PeakonCreation3} shows how the peak moves in space-time.

\begin{figure}
\begin{center}
\begin{tikzpicture}
\draw[->] (-1.5, 0) -- (4.1, 0) node[right] {\(x\)}; 
\draw[->] (0, -0.5) -- (0, 2.2) node[above] {\(u(x, t)\)};
\draw [domain = -1.3 : -0.355] plot (\x, {(exp(\x) +  exp(-\x))*exp(-0.2))}); 
\draw [domain = -0.355 : 4.0] plot (\x, {exp(-\x + 0.2)}); 
\end{tikzpicture}
\caption{Wave profile of \(u_5\), shortly after the time of creation}
\label{fig:PeakonCreation1}
\end{center}
\end{figure}
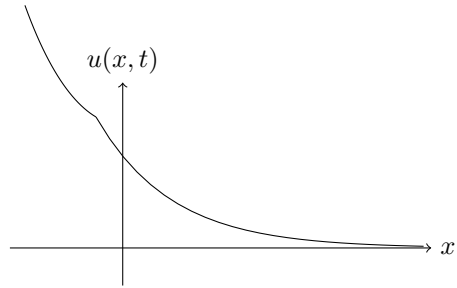

\begin{figure}
\begin{center}
\begin{tikzpicture}[scale=0.80]
\draw[->] (-2, 0) -- (4.1, 0) node[right] {\(x\)}; 
\draw[->] (0, -0.5) -- (0, 1.5) node[above] {\(u(x, t)\)};
\draw [domain = -2.0 : 0.687] plot (\x, {(exp(\x) +  exp(-\x))*exp(-0.8))}); 
\draw [domain = 0.687 : 4.0] plot (\x, {exp(-\x + 0.8)}); 
\end{tikzpicture} \,\,\,\,\,
\begin{tikzpicture}[scale=0.80]
\draw[->] (-4, 0) -- (4.1, 0) node[right] {\(x\)}; 
\draw[->] (0, -0.5) -- (0, 1.5) node[above] {\(u(x, t)\)};
\draw [domain = -4.0 : 2.798] plot (\x, {(exp(\x) +  exp(-\x))*exp(-2.8))}); 
\draw [domain = 2.798 : 4.0] plot (\x, {exp(-\x + 2.8)}); 
\end{tikzpicture}
\caption{Wave profile of \(u_5\), snapshots at two different later times}
\label{fig:PeakonCreation2}
\end{center}
\end{figure}
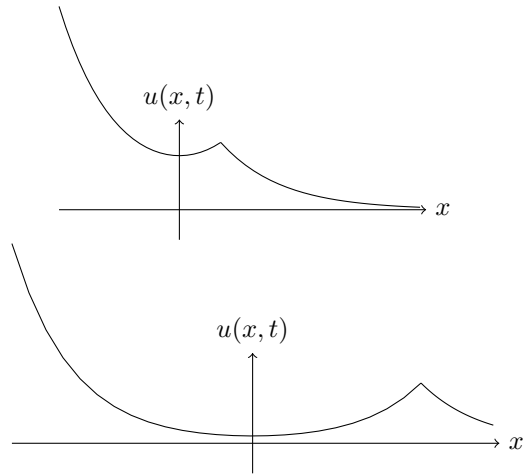

\begin{figure}
\begin{center}
\begin{tikzpicture}
\draw[->] (-2.5, 0) -- (2.5, 0) node[right] {\(x\)}; 
\draw[->] (0, -0.5) -- (0, 2.5) node[above] {\(t\)};
\draw[dashed, gray] (-2.5, 0.3) -- (2.5, 0.3);
\draw [domain = 0.34 : 2.5] plot ({ ln((exp(\x - 0.3) -  exp(-\x + 0.3))))}, \x);
\draw (1.9, 2) node[right] {\(x = B(t)\)}; 
\draw (-0.1, 0.3) -- (0.1, 0.3) node[above right] {\(t_0\)};
\end{tikzpicture}
\caption{Movement of the peak in space-time}
\label{fig:PeakonCreation3}
\end{center}
\end{figure}
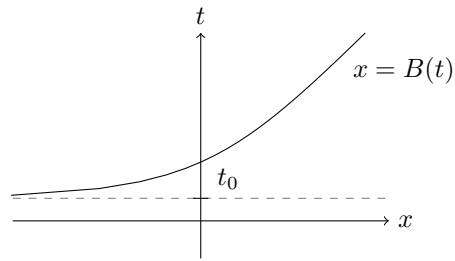

Let us also briefly consider the function \(u_4(x,t)\). By modifying the argument above, one gets that this function also has a peak, but \emph{before} a certain (finite) time, at which the position of the peak goes to \(+ \infty\). One can also check that \(u_4\) is a weak solution until the peak is destroyed, after which it is a regular solution to the Novikov equation.

Finally, let us mention what happens if one combines the transformations above. Applying transformation 5 with parameter \(\varepsilon\), followed by transformation 4 with parameter \(\delta\), gives the following function:
\[\tilde{u} = c \sqrt{1 + 2\delta e^{2x}} \sqrt{1 + 2\varepsilon(e^{-2x} + 2\delta)} e^{-\left|\frac{1}{2}\textnormal{ln}\left(\frac{1}{e^{-2x} + 2\delta} + 2\varepsilon\right) - c^2t\right|}.\]
It turns out that this function has a peak that is both created and destroyed in finite time. The precise interval for which the peak lives is 
\[t \in \left( \frac{1}{2c^2}\textnormal{ln}(2\varepsilon), \frac{1}{2c^2}\textnormal{ln}\Big(2\varepsilon + \frac{1}{2\delta}\Big)\right).\]
Outside this interval, \(\tilde{u}\) is a smooth function of \(x\), and thus a regular solution as before. To find a function for which the peak lives between given times \(t_1\) and \(t_2\), choose
\[\begin{cases}
\varepsilon = \frac{1}{2}e^{2c^2t_1}, \\
\delta = \frac{1}{2\left(e^{2c^2t_2} - e^{2c^2t_1}\right)}, \end{cases}
\quad t_1 < t_2.\]

\section{Peakon Creation in Related Equations}

Finding unbounded solutions with peakon creation in the Novikov equation inspires us to look for solutions with similar behaviour in the related Camassa--Holm and Degasperis--Procesi equations.

\subsection{Camassa--Holm solutions with peakon creation}

The Camassa--Holm equation (CH), from \cite{CHOriginal}, is given by
\begin{equation}u_t - u_{xxt} + 3uu_x = 2u_xu_{xx} + uu_{xxx}.\label{CH} \end{equation}
It is known from \cite{CHSymmetries} that the CH symmetry group only consists of translations and scalings. This means that we cannot find solutions with peakon creation just by transforming the one-peakon solution. Still, it turns out that there are solutions with peakon creation, that one can find via an ansatz inspired by the Novikov solutions found in the previous section.
\begin{Theorem}
For every \(t \in \mathbb{R}\), the function
\[u(x, t) = \begin{cases}
u^+ = a(t)e^{-x}, &x \geq B(t), \\
u^- = c(t)(e^x + e^{-x}), &x \leq B(t), \end{cases} \]
where 
\begin{gather*}
a(t) = U\textnormal{cosh}[U(t - t_0)], \\
B(t) = \textnormal{ln}(\textnormal{sinh}[U(t - t_0)]), \\
c(t) = \frac{U}{\textnormal{cosh}[U(t - t_0)]},
\end{gather*}
is a solution to the Camassa--Holm equation.
\end{Theorem}
(Note that for \(t \leq t_0\), \(B(t)\) is undefined, so we take \(u(x,t) = u^+\) for those times.)
\begin{proof}
We look for weak solutions of the kind
\begin{equation}
u(x, t) = \begin{cases}
u^+ = a(t)e^{-x}, &x \geq B(t), \\
u^- = c(t)(e^x + e^{-x}), &x \leq B(t), \end{cases}
\end{equation}
where \(a(t)\) and \(c(t)\) are positive continuous functions, chosen in such a way that \(u\) is continuous at the peak \(B(t)\) for all times. From the weak formulation of the Camassa--Holm equation found in \cite{NovikovPeakons}, one has that \(u\) must satisfy
\begin{equation}
\left< \left(1 - \partial_x^2\right)u_t + \left(3 - \partial_x^2\right)\partial_x\left(\frac{1}{2}u^2\right) + \partial_x\left(\frac{1}{2}u_x^2\right) , \phi \right> = 0
\end{equation}
for all test functions \(\phi(x) \in C_0^{\infty}\). Note that the function \(u(x, t)\) is a strong solution of (\ref{CH}) on each interval. Thus integration by parts, as in the previous section, gives that
\begin{gather*} U_1(B)\phi(B) + U_2(B)\phi_x(B) + U_3(B)\phi_{xx}(B) = 0 \end{gather*}
must be satisfied, where
\begin{subequations}
\begin{align} U_1(B) &:= \left(u_t^-\right)_x(B) - \left(u_t^+\right)_x(B) + \frac{1}{2}\left(\left(u^-\right)^2\right)_{xx}(B) - \frac{1}{2}\left(\left(u^+\right)^2\right)_{xx}(B) + \nonumber \\
&+ \frac{1}{2}\left(u_x^+(B)\right)^2 - \frac{1}{2}\left(u_x^-(B)\right)^2 + \frac{3}{2}\left(\left(u^+\right)^2\right)(B) - \frac{3}{2}\left(\left(u^-\right)^2\right)(B), \label{CHU1} \\
U_2(B) &:= u_t^+(B) - u_t^-(B) + \frac{1}{2}\left(\left(u^+\right)^2\right)_x(B) - \frac{1}{2}\left(\left(u^-\right)^2\right)_x(B), \label{CHU2} \\
U_3(B) &:= \frac{1}{2}\left(\left(u^-\right)^2\right)(B) - \frac{1}{2}\left(\left(u^+\right)^2\right)(B). \label{CHU3} \end{align} \end{subequations}
The condition (\ref{CHU3}) \(= 0\) is met because of continuity. Using continuity, we can also express \(a(t)\) in terms of \(B\) and \(c\), since
\[c\left(e^{-B} + e^B\right) = ae^{-B} \implies a = c\left(1 + e^{2B}\right) \implies \frac{da}{dt} = \frac{dc}{dt}\left(1 + e^{2B}\right) + 2\frac{dB}{dt}ce^{2B}.\]
Eliminating \(a\) and its time derivative in the conditions (\ref{CHU1}) \(=\) (\ref{CHU2}) \(= 0\) gives the system
\begin{subequations}
\label{CHODESystem}
\begin{gather}
\frac{d}{dt}\left(ce^B \right) = c^2, \\
\frac{dB}{dt} = c\left(e^B + e^{-B}\right).
\end{gather}
\end{subequations}
 These conditions are simplified by a change of variables,
\[\begin{cases}
G(t) = c(t)e^{B(t)}, \\
K(t) = \frac{1}{c^2(t)}, \end{cases} \implies
\begin{cases}
\frac{dG}{dt} = c^2 = \frac{1}{K}, \\
\frac{dK}{dt} = \frac{-2}{c^3} \frac{dc}{dt} = 2KG, \end{cases} \]
where the last line follows from the observation that
\[\frac{dc}{dt} = \frac{d}{dt}\left( \frac{G}{e^B} \right) = \frac{c^2}{e^B} - \frac{Gc\left(e^B + e^{-B}\right)e^B}{e^{2B}} = -c^2e^B = -cG.\]
One can now get a separable differential equation and find a constant of motion:
\[\frac{dK}{dG} = \frac{\frac{dK}{dt}}{\frac{dG}{dt}} = 2K^2G \implies \int \frac{dK}{K^2} = \int 2G \, dG \implies -\frac{1}{K} = G^2 + \textnormal{constant}.\]
Apart from the trivial solution \(a(t) = c(t) = 0\), \(G\) and \(\frac{1}{K}\) are positive, so the constant has to be negative.  Let the constant be named \(-U^2\) for convenience. Then
\[\frac{dG}{dt} = \frac{1}{K} = U^2 - G^2 \implies \int \frac{dG}{U^2 - G^2} = \int dt\]
\[\implies \frac{1}{2U}\int \left( \frac{1}{U + G} + \frac{1}{U - G} \right) dU = \int dt \implies \frac{1}{2U} \, \textnormal{ln} \left( \frac{U + G}{U - G} \right) = t - t_0\]
\[\implies G = U\frac{e^{2U(t - t_0)} - 1}{e^{2U(t - t_0)} + 1} = U \textnormal{tanh}[U(t - t_0)].\]
From this one gets \(K\) as
\[K = \frac{1}{U^2 - G^2} = \frac{1}{U^2}\cdot\frac{1}{1 - \textnormal{tanh}^2[U(t - t_0)]} = \frac{\textnormal{cosh}^2[U(t - t_0)]}{U^2},\]
which gives expressions for \(c(t)\), \(B(t)\), and consequently \(a(t)\):
\begin{gather*}
c(t) = \frac{1}{\sqrt{K}} = \frac{U}{\textnormal{cosh}[U(t - t_0)]}, \\
B(t) = \textnormal{ln}\left(G\sqrt{K}\right) = \textnormal{ln}(\textnormal{sinh}[U(t - t_0)]), \\
a(t) = c(t)\left(1 + e^{2B(t)}\right) = \frac{U}{\textnormal{cosh}[U(t - t_0)]}\left(1 + \textnormal{sinh}^2[U(t - t_0)]\right) = U\textnormal{cosh}[U(t - t_0)].
\end{gather*}
\end{proof}
We note that our new solution behaves similarly to the Novikov solution with peakon creation in \hyperref[NovikovPeakonCreationSolution]{Theorem \ref{NovikovPeakonCreationSolution}}. Up to time \(t_0\), the expression for \(B(t)\) is undefined, so the function is a strong solution to the Camassa--Holm equation. At \(t_0\) a peak is created at \(x = -\infty\), which then moves rapidly in from the left. Note that the exact time dependencies are not the same as for the Novikov peakon-creation solution, even though the qualitative behaviour is the same.

\subsection{Degasperis--Procesi solutions with peakon creation?}

The Degasperis--Procesi (DP) equation \cite{DPOriginal} is given by
\begin{equation}u_t - u_{xxt} + 4uu_x = 3u_xu_{xx} + uu_{xxx}.\label{DP} \end{equation}
Like Camassa--Holm, it only has scaling and translation symmetries \cite{DPSymmetries}, so we try to find peakon-creation solutions using the same method as in the last section. 

\begin{Theorem}
For every \(t \in \mathbb{R}\), the function
\[u(x, t) = \begin{cases}
u^+ = a(t)e^{-x}, &x \geq B(t), \\
u^- = c(t)(e^x + e^{-x}), &x \leq B(t), \end{cases}\]
where 
\begin{gather*}
a(t) = \sqrt{\frac{C_1}{C_0}} \left(\frac{1 + C_0C_1e^{2Ut}}{e^{Ut} + \frac{e^{-Ut}}{UC_0}}\right), \\
B(t) = \textnormal{ln}\sqrt{C_0C_1} + Ut, \\
c(t) = \sqrt{\frac{C_1}{C_0}} \frac{1}{e^{Ut} + \frac{e^{-Ut}}{UC_0}},
\end{gather*}
is a solution to the Degasperis--Procesi equation.
\end{Theorem}
\begin{proof}
We look for weak solutions
\[u(x, t) = \begin{cases}
a(t)e^{-x}, &x \geq B(t), \\
c(t)\left(e^x + e^{-x}\right), &x \leq B(t), \end{cases}\]
where \(a(t)\) and \(c(t)\) are positive continuous functions, such that \(u\) is continuous at the peak \(B(t)\) for all times. We stick to the weak formulation given in \cite{NovikovPeakons}, i.e., \(u(x,t)\) must satisfy
\[\left< \left(1 - \partial_x^2\right)u_t + \left(4 - \partial_x^2\right)\partial_x\left(\frac{1}{2}u^2\right), \phi \right> = 0.\]
As before, \(u_t\) is piecewise continuous, so via integration by parts we find three conditions on \(u^+\) and \(u^-\) at the peak, one of which is satisfied because of continuity. Eliminating \(a(t)\), we end up with a system similar to (\ref{CHODESystem}), but not the same:
\begin{gather*}
\frac{d}{dt}\left(ce^B \right) = 2c^2, \\
\frac{dB}{dt} = c\left(e^B + e^{-B}\right).
\end{gather*}
With \(G(t) = c(t)e^{B(t)}\), \(K(t) = \frac{e^{B(t)}}{c(t)}\), we get
\begin{gather*}
\frac{dG}{dt} = \frac{2G}{K}, \\
\frac{dK}{dt} = 2GK.
\end{gather*}
Using the same method as before, we find a relation between \(K\) and \(G\):
\[\frac{dK}{dG} = \frac{\frac{dK}{dt}}{\frac{dG}{dt}} = K^2 \implies \int \frac{dK}{K^2} = \int \, dG \implies -\frac{1}{K} = G + \textnormal{constant}.\]
Let the constant be named \(-U\). Since \(G\) and \(\frac{1}{K}\) are nonnegative, \(U = 0\) only gives the trivial solution \(a(t) = c(t) = 0\). Assume \(U \neq 0\). Then
\[\frac{dK}{dt} = 2GK = 2K\left(U - \frac{1}{K}\right) \implies \frac{dK}{dt} - 2KU = -2,\]
which has the general solution
\[K = C_0e^{2Ut} + \frac{1}{U}.\]
This gives \(G(t)\) via
\[\frac{dG}{dt} = \frac{2G}{K} = \frac{2G}{C_0e^{2Ut} + \frac{1}{U}} \implies G = \frac{C_1}{\frac{e^{-2Ut}}{UC_0} + 1},\]
so we get
\[e^{B(t)} = \sqrt{GK} = \sqrt{C_1}\sqrt{\frac{C_0e^{2Ut} + \frac{1}{U}}{\frac{e^{-2Ut}}{UC_0} + 1}} = \sqrt{C_0C_1e^{2Ut}} = \sqrt{C_0C_1}e^{Ut}\]
\[\implies B(t) = \textnormal{ln}\sqrt{C_0C_1} + Ut,\]
and 
\[c(t) = \sqrt{\frac{G}{K}} = \sqrt{\frac{C_1}{\left(\frac{e^{-2Ut}}{UC_0} + 1\right)\left(C_0e^{2Ut} + \frac{1}{U}\right)}} =\]
\[= \sqrt{\frac{C_1}{C_0e^{2Ut}\left(1 + \frac{e^{-2Ut}}{UC_0}\right)^2}} = \sqrt{\frac{C_1}{C_0}} \frac{1}{e^{Ut} + \frac{e^{-Ut}}{UC_0}}.\]
This gives
\[a(t) = c(t)\left(1 + e^{2B(t)}\right) = \sqrt{\frac{C_1}{C_0}} \left(\frac{1 + C_0C_1e^{2Ut}}{e^{Ut} + \frac{e^{-Ut}}{UC_0}}\right).\]
\end{proof}

Note that \(B(t)\) here is defined for all times, so there is no peakon creation in this solution. We have found an unbounded piece-wise defined solution though. It is possible that a more general ansatz yields a solution with peakon creation in the DP case. It would also be interesting to investigate if one can find a solution with creation of so-called shockpeakons \cite{DPShockpeakons}.

\appendix

\renewcommand\thesection{\Alph{section}}

\section{Lie Symmetries}

In this appendix we use the framework of symmetry groups, due to Lie, to construct transformations taking solutions of the Novikov equation (\ref{eq:Novikov}) to other solutions. Similar results have been presented for the related Camassa--Holm equation in \cite{CHSymmetries} and more recently for the Degasperis--Procesi equation in \cite{DPSymmetries}. Note that computation of symmetry groups is quite cumbersome, so to find them explicitly, the Jets package in Maple is used. For more information on the Jets algorithm and how to use the package, see \cite{JetsAlgorithm} and \cite{JetsGuide} respectively.

\subsection{Definitions}

Herein we will mainly use the notation employed in Olver's book \cite{Olver}, which also contains all details and proofs omitted in this section.

Let \(X = \{\bar x= \left(x^1, \dots, x^p\right)\}\) and \(U = \{\bar u = \left(u^1, \dots, u^q\right)\}\) be the spaces of independent and dependent variables, respectively, involved in a system of differential equations. The \emph{n-th prolongation} of a scalar function \(u\) is defined as a tuple, denoted \(u^{(n)}\), containing \(u\) and all its derivatives up to order \(n\), where derivatives are arranged by order and then lexicographically. For example, with independent variables \(x^1 = x, x^2 = t\) one gets \(u^{(2)} = \left(u, u_x, u_t, u_{xx}, u_{xt}, u_{tt}\right)\). Furthermore, we define for vector-valued functions
\[\bar u^{(n)} = \left((u^1)^{(n)}, \dots, (u^q)^{(n)}\right),\]
and set \(U^{(n)} = \{\bar u^{(n)} \mid \bar u \in U\}\).

An \(n\)-th order system of differential equations can then be given as
\begin{equation}
\Delta_r\left(\bar x, \bar u^{(n)}\right) = 0, \quad r = 1, \dots, l, \label{DeltaSystem}
\end{equation}
where the system has \emph{maximal rank} if the Jacobian \(J_\Delta \left(\bar x, \bar u^{(n)}\right)\) has rank~\(l\) for all points \(\left(\bar x, \bar u^{(n)}\right)\) that are solutions to the system.

If \(G\) is a local group of transformations on \(M \subset X \times U\) and \(g \in G\), one defines the prolonged action \(g^{(n)}\) on a point \(\left(\bar x, \bar u^{(n)}\right) \in M^{(n)} \subset X \times U^{(n)}\) as transforming \(\bar x\) and \(\bar u\), and then re-evaluating derivatives. What we are looking for are \emph{symmetry groups}, i.e., local groups of transformations on \(M\) such that their prolongations take solutions of the system (\ref{DeltaSystem}) to other solutions. 

To a one-parameter group \(G\) there corresponds an \emph{infinitesimal generator} \(\textbf{v}\), which is a vector field defined on \(M\), with the property that orbits of the group action are maximal integral curves of \(\textbf{v}\). Similarly, to an \(m\)-parameter group there corresponds a set of \(m\) infinitesimal generators \(\textbf{v}_1, \dots, \textbf{v}_m\), which has the property that it is closed under taking Lie bracket, and that each infinitesimal generator corresponds to a generator of the group \(G\). 

We define the prolongation of an infinitesimal generator \(\textbf{v}\) of a group \(G\) to be the vector field \(\textbf{v}^{(n)}\), defined on \(M^{(n)}\), which is the infinitesimal generator of the group \(G^{(n)} := \{g^{(n)} | g \in G\}\). We want to give a formula for computing \(\textbf{v}^{(n)}\).

Let \(J\) be a multi-index of the form 
\[J = (j_1, \dots, j_k), \quad 1 \leq j_k \leq p, \quad 1 \leq k \leq n,\]
where \(p\) is the number of independent variables. Then one can introduce a compact notation for derivatives as 
\[u^{\alpha}_j =  \frac{\partial u^{\alpha}}{\partial x_j} \quad \textnormal{and} \quad u^{\alpha}_J =  \frac{\partial^k u^{\alpha}}{\partial x_{j_1} \dotsm \partial x_{j_k}},\]
and we shall also use the notation
\[D_j \phi(\bar x, \bar u) = \frac{\partial \phi}{\partial x^j} + \sum_{\alpha = 1}^q u^{\alpha}_j \frac{\partial \phi}{\partial u^{\alpha}}\] 
for total derivatives, and \(D_J = D_{j_1}D_{j_2} \dotsm D_{j_k}\) for multi-indices \(J\).

The following theorem (Theorem 2.36 in \cite{Olver}) gives the general formula for~\(\textbf{v}^{(n)}\):
\begin{Theorem} \label{Th:1}
Let
\begin{equation}
\textnormal{\textbf{v}} = \sum_{i=1}^p \xi^i(\bar x,\bar u) \frac{\partial}{\partial x^i} + \sum_{\alpha = 1}^q \phi_{\alpha}(\bar x,\bar u) \frac{\partial}{\partial u^{\alpha}} \label{VectorFieldDef}
\end{equation}
be a vector field on \(M \subset X \times U\). Then
\begin{equation}
\textnormal{\textbf{v}}^{(n)} = \textnormal{\textbf{v}} + \sum_{\alpha = 1}^q \sum_J \phi_{\alpha}^J \left(\bar x, \bar u^{(n)}\right) \frac{\partial}{\partial u_J^{\alpha}}, \label{VectorFieldProlongation} \end{equation}
where the second sum is over all multi-indices \(J\), and \(\phi_{\alpha}^J\) is given by
\begin{equation}
\phi_{\alpha}^J \left(\bar x, \bar u^{(n)}\right) = D_J \left( \phi_{\alpha} - \sum_{i = 1}^p \xi^i \frac{\partial u^{\alpha}}{\partial x^i} \right) + \sum_{i=1}^p \xi^i \frac{\partial u^{\alpha}_J}{\partial x^i}. \label{VectorFieldComponents} \end{equation}
\end{Theorem}

The next theorem (Theorem 2.31 in \cite{Olver}) is the main tool for finding symmetry groups:
\begin{Theorem} \label{Th:2}
Suppose
\[\Delta_r\left(\bar x, \bar u^{(n)}\right) = 0, \quad r = 1, \dots, l, \]
is a system of differential equations of maximal rank defined over \(M \subset X \times U\). If \(G\) is a local group of transformations acting on \(M\), with infinitesimal generator \textnormal{\textbf{v}}, and
\[\textnormal{\textbf{v}}^{(n)} \left( \Delta_r\left(\bar x, \bar u^{(n)}\right) \right) = 0, \quad r = 1, \dots, l, \quad whenever \quad \Delta\left(\bar x, \bar u^{(n)}\right) = 0,\]
then \(G\) is a symmetry group of the system.
\end{Theorem}

Thus, the method for finding symmetry groups is to make the ansatz (\ref{VectorFieldDef}) for \(\textbf{v}\), prolong it using expressions (\ref{VectorFieldProlongation}) and (\ref{VectorFieldComponents}), apply it as a differential operator to the system (\ref{DeltaSystem}), and find the conditions for which this expression is zero. Then \(\textbf{v}\) is an infinitesimal generator of the symmetry group, so finding \(G\) is just a matter of exponentiating the vector field.

\subsection{Using Jets}

The computations required to determine \(\textbf{v}\) become increasingly more involved as the number of variables or the number of equations in the system grows. A semi-automatic process, called Jets, is used here to solve this problem. Jets is implemented in Maple, and it is well suited for dealing with large symbolic expressions appearing in the ansatz for \(\textbf{v}^{(n)}\). More concretely, what happens is the following:

Let \(\textbf{v}\) be defined as in (\ref{VectorFieldDef}). As a computational trick, define
\[Q^{\alpha}\left(\bar x, \bar u^{(1)}\right) = \phi_{\alpha}(\bar x, \bar u) - \sum_{i=1}^p \xi^i(\bar x, \bar u) u_i^{\alpha}, \quad \alpha = 1, \dots, q.\]
We call \(Q = (Q^1, \dots, Q^q)\) the \emph{characteristic} of \(\textbf{v}\). Note that one can recover \(\textbf{v}\) from \(Q\) through the relations
\begin{equation} \label{RecoverVFromQRelations}
\begin{cases}
\xi^i(\bar x, \bar u) = - \frac{\partial}{\partial u_i^{\alpha}} Q^{\alpha}, \\
\phi_{\alpha}(\bar x, \bar u) = Q^{\alpha}\left(\bar x, \bar u^{(1)}\right) + \sum_{i=1}^p \xi^i(\bar x, \bar u) u_i^{\alpha}.
 \end{cases}\end{equation}
Jets is built to produce \(Q\), so that we can recover \(\textbf{v}\) and exponentiate it to find the symmetry group.

The Novikov equation, as stated before, is
\[u_{xxt} - u_t = 4u^2u_x - 3uu_xu_{xx} - u^2u_{xxx}.\]
We note that this is just a single third-order partial differential equation, with two independent and one dependent variable. This means that one can drop the \(\alpha\)'s and the bar on \(\bar u\) in the equations above. Also, let \(x^1 = x\), \(x^2 = t\), so that the ansatz for \(\textbf{v}\) becomes
\[\textbf{v} = \xi^x(x, t, u) \frac{\partial}{\partial x} + \xi^t(x, t, u) \frac{\partial}{\partial t}+ \phi(x, t, u) \frac{\partial}{\partial u},\]
and its third prolongation 
\begin{align*}
\textbf{v}^{(3)} =  \textbf{v} &+ \phi^x\frac{\partial}{\partial u_x} + \phi^t\frac{\partial}{\partial u_t} + \phi^{xx}\frac{\partial}{\partial u_{xx}} + \phi^{xt}\frac{\partial}{\partial u_{xt}} + \phi^{tt}\frac{\partial}{\partial u_{tt}} + \\
&+ \phi^{xxx}\frac{\partial}{\partial u_{xxx}} + \phi^{xxt}\frac{\partial}{\partial u_{xxt}} + \phi^{xtt}\frac{\partial}{\partial u_{xtt}} + \phi^{ttt}\frac{\partial}{\partial u_{ttt}}.
\end{align*}

If one wanted to do the work manually one would now compute the coefficients \(\phi^x\), etc., using \hyperref[Th:1]{Theorem \ref{Th:1}}, apply \(\textbf{v}^{(3)}\) to the Novikov equation, and find conditions on the \(\xi\)'s and \(\phi\). Instead, let's go with Jets, and study the characteristic
\[Q(x, t, u, u_x, u_t) = \phi(x, t, u) - \xi^x(x, t, u) u_x - \xi^t(x, t, u) u_t.\]

With the following setup, Jets will generate all conditions for \(Q\) being the characteristic of the Novikov equation:
\begin{verbatim}
> read("Jets.s");
> coordinates([x,t], [u], 3);
> equation ('u_xxt' = u_t + 4*u^2*u_x - 3*u*u_x*u_xx - u^2*u_xxx);
> S := symmetries(u = Q);
> dependence(Q(x, t, u, u_t, u_x));
> unknowns(Q);
> run(S);
> dependence();
> S1 := clear(pds);
\end{verbatim}

We find that Q depends on all variables in general, and must satisfy the following conditions:
\begin{subequations}
\begin{gather}
\frac{\partial^2}{\partial t^2} Q = \frac{\partial^2}{\partial u_x^2} Q = \frac{\partial^2}{\partial u_t^2} Q = 0, \label{NovikovQConditionsA} \\
\frac{\partial^2}{\partial t \partial x} Q = \frac{\partial^2}{\partial u_x \partial t} Q = \frac{\partial^2}{\partial u_t \partial x} Q = \frac{\partial^2}{\partial u_t \partial u_x} Q = 0, \label{NovikovQConditionsB} \\
\left[\frac{\partial^2}{\partial u_t \partial t} - \frac{1}{u_t}\frac{\partial}{\partial t} \right] Q  = 0, \\
\left[\frac{\partial}{\partial u} + \frac{1}{u}\left(u_x \frac{\partial}{\partial u_x} + u_t \frac{\partial}{\partial u_t} - 1\right) \right] Q = 0, \\
\left[\frac{\partial^2}{\partial u_x \partial x} + \frac{2}{u}\left(1 - u_x \frac{\partial}{\partial u_x} - u_t \frac{\partial}{\partial u_t}\right) - \frac{1}{u_t}\frac{\partial}{\partial t}\right] Q  = 0, \\
\left[\frac{\partial^2}{\partial x^2} + \frac{2u_x}{u}\frac{\partial}{\partial x} + \frac{2(u^2 - u_x^2)}{u u_t}\frac{\partial}{\partial t} + \frac{4(u^2 - u_x^2)}{u^2}\left(u_x\frac{\partial}{\partial u_x} + u_t\frac{\partial}{\partial u_t} - 1\right) \right] Q = 0.
\end{gather}
\end{subequations}
It follows from (\ref{NovikovQConditionsA}) and (\ref{NovikovQConditionsB}) that the characteristic \(Q\) is a polynomial of first degree in both \(u_x\) and \(t\), with no mixed terms, so one can split it into three parts, denoted \(Q_0\), \(Q_1\) and \(Q_2\), that only depend on \(u\), \(x\) and \(u_t\), so that \(Q\) = \(Q_0 u_x + Q_1 t + Q_2\). This simplifies the dependence of Q, so we run Jets again:
\begin{verbatim}
> Q := Q0*u_x + Q1*t + Q2;
> dependence(Q0(u, x, u_t), Q1(u, x, u_t), Q2(u, x, u_t));
> unknowns(Q0, Q1, Q2);
> run(S1);
> dependence();
> S2 := clear(pds);
\end{verbatim}

This time, Jets is able to reduce the dependencies, so that \(Q_0\) now only depends on \(x\), while \(Q_1\) only depends on \(u_t\). However, \(Q_2\) still depends on \(u\), \(x\), and \(u_t\). The list of conditions is now more manageable:
\begin{subequations}
\begin{gather}
\left(\frac{\partial^3}{\partial x^3} - 4\frac{\partial}{\partial x}\right) Q_0 = 0, \label{NovikovQConditions2A} \\
\left(\frac{\partial}{\partial u_t} - \frac{1}{u_t}\right) Q_1 = 0, \label{NovikovQConditions2B}\\
\frac{u}{2}\frac{\partial^2}{\partial x^2}Q_0 + \frac{\partial}{\partial x}Q_2 = 0, \label{NovikovQConditions2C} \\
\frac{1}{2}\frac{\partial}{\partial x}Q_0 - \frac{1}{2u_t}Q_1 + \frac{\partial}{\partial u}Q_2 = 0, \\
-\frac{u}{2u_t}\frac{\partial}{\partial x}Q_0 + \frac{u}{2u_t^2}Q_1 + \left(\frac{\partial}{\partial u_t} - \frac{1}{u_t}\right)Q_2 = 0. \label{NovikovQConditions2E}
\end{gather}
\end{subequations}
Now, conditions (\ref{NovikovQConditions2A}) and (\ref{NovikovQConditions2B}) imply that
\begin{gather*}
Q_0 = Q_{00}e^{2x} + Q_{01}e^{-2x} + Q_{02}, \\
Q_1 = Q_{10}u_t, \end{gather*}
where \(Q_{00}\) up to \(Q_{10}\) are constants. Inserting these expressions into conditions (\ref{NovikovQConditions2C}) through (\ref{NovikovQConditions2E}) and solving for \(Q_2\) gives
\[Q_2 = -u Q_{00}e^{2x} + u Q_{01} e^{-2x} + \frac{u}{2}Q_{10} + u_t Q_{20},\]
where \(Q_{20}\) is also constant.

We conclude that the most general characteristic for the Novikov equation is
\begin{equation}
Q = \left(-u e^{2x} + u_x e^{2x}\right)Q_{00} + \left(u e^{-2x} + u_x e^{-2x}\right)Q_{01} + u_x Q_{02} + \left(\frac{1}{2}u + t u_t\right)Q_{10} + u_t Q_{20}.
\label{NovikovCharacteristic} \end{equation}
Note that it has five degrees of freedom, which correspond to five different generators for the symmetry group. From the characteristic, we recover five infinitesimal generators, using (\ref{RecoverVFromQRelations}).
\begin{gather*}
\textbf{v}_1 = -\frac{\partial}{\partial x}, \\
\textbf{v}_2 = -\frac{\partial}{\partial t}, \\
\textbf{v}_3 = -\frac{\partial}{\partial t} + \frac{u}{2}\frac{\partial}{\partial u}, \\
\textbf{v}_4 = -e^{2x}\frac{\partial}{\partial x} -e^{2x}u\frac{\partial}{\partial u}, \\
\textbf{v}_5 = -e^{-2x}\frac{\partial}{\partial x} +e^{-2x}u\frac{\partial}{\partial u}. \end{gather*}
Exponentiating the vector fields, we find the symmetry group of the Novikov equation.
\begin{Theorem} \label{Th:3}
If \(u = f(x, t)\) solves the Novikov equation (\ref{eq:Novikov}), then so do
\begin{gather*}
u_1 = f(x - \varepsilon, t), \\
u_2 = f(x, t - \varepsilon), \\
u_3 = e^{\varepsilon/2} f(x, te^{\varepsilon}), \\
u_4 = \sqrt{1 + 2\varepsilon e^{2x}} f\left(-\frac{1}{2} \textnormal{ln}\left(e^{-2x} +2\varepsilon\right), t\right), \\
u_5 = \sqrt{1 + 2\varepsilon e^{-2x}} f\left(\frac{1}{2} \textnormal{ln}\left(e^{2x} +2\varepsilon\right), t\right).\end{gather*}
\end{Theorem} It is easy to check the first three by inspecting the equation; the last two are best checked by computer.

Finally, while computing the Lie symmetries of the Novikov equation, we also did the same for its two-component generalization due to Geng--Xue \cite{GXOriginal}. While not directly relevant to this article, this might be a good place to mention the results. The Geng--Xue system is given by
\[\begin{cases}
\label{eq:GX}
u_{xxt} - u_t = (u_x - u_{xxx})uv + 3(u - u_{xx})vu_x, \\
v_{xxt} - v_t = (v_x - v_{xxx})uv + 3(v - v_{xx})uv_x. \end{cases}\]
Proceeding with the help of Jets as before, we find the following symmetries.

\begin{Theorem}
If \[\begin{cases}
u = f(x, t), \\
v = g(x, t), \end{cases}\] solves the Geng--Xue system (\ref{eq:GX}), then so do
\begin{align*}
&\begin{cases}
u_1 = \sqrt{1 + 2\varepsilon e^{2x}} f\left(-\frac{1}{2} \textnormal{ln}\left(e^{-2x} +2\varepsilon\right), t\right), \\
v_1= \sqrt{1 + 2\varepsilon e^{2x}} g\left(-\frac{1}{2} \textnormal{ln}\left(e^{-2x} +2\varepsilon\right), t\right), \end{cases} \\
&\begin{cases}
u_2 = \sqrt{1 + 2\varepsilon e^{-2x}} f\left(\frac{1}{2} \textnormal{ln}\left(e^{2x} +2\varepsilon\right), t\right), \\
v_2 = \sqrt{1 + 2\varepsilon e^{-2x}} g\left(\frac{1}{2} \textnormal{ln}\left(e^{2x} +2\varepsilon\right), t\right), \end{cases} \\
&\begin{cases}
u_3 = f(x - \varepsilon, t), \\
v_3 = g(x - \varepsilon, t), \end{cases}
\begin{cases}
u_4 = f(x, t - \varepsilon), \\
v_4 = g(x, t - \varepsilon), \end{cases} \\ 
&\begin{cases}
u_5 = f(x, te^{\varepsilon}), \\
v_5 = e^{\varepsilon}g(x, te^{\varepsilon}), \end{cases}
\, \begin{cases}
u_6 = e^{\varepsilon}f(x, t), \\
v_6 = e^{-\varepsilon}g(x,t). \end{cases}
\end{align*}
\end{Theorem}

\section*{Acknowledgements}
The author would like to thank the organizers of the Silesian Mathematical Summer School on Geometry of Differential Equations, given 17--21 September 2012 at the Silesian University, Opava, Czech Republic, for giving the theoretical background and practical demonstration on how to use the Jets algorithm.

Thanks also to Hans Lundmark, Stefan Rauch and Joakim Arnlind for reading the manuscript and providing valuable comments and suggestions for improvements.

%
%
%
%
%
%
%
%
%
%
%

\end{document}